\documentclass[twocolumn,secnumarabic,amssymb, nofootinbib,aps, superscriptaddress, prl]{revtex4-1}

\usepackage{amsthm}
\usepackage{amsmath}
\usepackage{color}
\usepackage{adjustbox}
\usepackage{graphicx}
\usepackage{hyperref}

\newcommand{\ket}[1]{\left| #1\right. \rangle}

\newcommand{\dd}{\mathrm{d}}

\newcommand{\SuppMat}{\cite{SuppMat}}
\newtheorem{thm}{Theorem}
\newtheorem*{thm*}{Theorem}
\newtheorem{lma}{Lemma}
\newtheorem*{lma*}{Lemma}

\begin{document}

\title{Arbitrarily small amount of measurement independence is sufficient to manifest quantum nonlocality}%

\author{Gilles P\"utz}%
\email[]{puetzg@unige.ch}
\affiliation{Group of Applied Physics, University of Geneva, CH-1211 Geneva 4, Switzerland.}
\author{Denis Rosset}
\affiliation{Group of Applied Physics, University of Geneva, CH-1211 Geneva 4, Switzerland.}
\author{Tomer Jack Barnea}
\affiliation{Group of Applied Physics, University of Geneva, CH-1211 Geneva 4, Switzerland.}
\author{Yeong-Cherng Liang}
\affiliation{Institute for Theoretical Physics, ETH Zurich, 8093 Zurich, Switzerland.}
\author{Nicolas Gisin}
\affiliation{Group of Applied Physics, University of Geneva, CH-1211 Geneva 4, Switzerland.}

\date{\today}%
\begin{abstract}
The use of Bell's theorem in any application or experiment relies on the assumption of free choice or, more precisely,  measurement independence, meaning that the measurements can be chosen freely. Here, we prove that even in the simplest Bell test --- one involving 2 parties each performing 2 binary-outcome measurements --- an \emph{arbitrarily small amount} of measurement independence is \emph{sufficient} to manifest quantum nonlocality. To this end, we introduce the notion of measurement dependent locality and show that the corresponding correlations form a convex polytope. These correlations can thus be characterized efficiently, e.g., using a finite set of Bell-like inequalities --- an observation that enables the systematic study of quantum nonlocality and related applications under limited measurement independence.
\end{abstract}

\maketitle

Since Bell's seminal work~\cite{Bell1964}, quantum nonlocality has gathered more and more interest, not only from a foundational point of view but also as a resource in several tasks like quantum key distribution~\cite{Ekert91,Acin06}, randomness expansion \cite{Colbeck2006,Pironio2010}, randomness extraction \cite{Bouda2014} or robust certification~\cite{Bancal11} and quantification~\cite{Moroder13} of quantum entanglement. It has led to the notion of device independence (see, eg.,~\cite{Brunner:RMP}), where the violation of a Bell inequality alone certifies properties that are useful to the task at hand, e.g., non-determinism of the outputs. In such a scenario, it is enough to consider black boxes that the parties give an input to and get an outcome from instead of having to consider the complex physical description of the implementation.

However, an important assumption has to be made in order for violations of Bell inequalities to exclude any local, and in particular deterministic explanation. Let us consider an adversarial scenario in which the boxes were in the hands of an adversary Eve before being given to the parties performing the experiment or protocol. The inputs for the boxes are chosen by local random number generators. If the adversary could influence these random number generators, then she can prepare the boxes with local strategies that appear to be nonlocal to the parties. The violation of a Bell-inequality therefore does not imply that the outcomes of the boxes are unknown to the adversary unless we assume that the inputs are independent of the adversary and that she cannot gain any information about them. This assumption is commonly referred to as measurement independence~\cite{Hall2011,Barrett2011}. Similar, but slightly stronger assumptions~\cite{Vona2014} are the assumptions of free choice~\cite{Colbeck2012} and free will~\cite{Bell2004}. Ensuring measurement independence in a Bell test seems impossible. However, if we abandon measurement independence completely and place no restriction at all on the adversary's influence, then it is impossible to show and exploit quantum nonlocality~\cite{Brans1988}.

In light of this, relaxations of this assumption have gathered attention and been studied in recent works. Hall~\cite{Hall2011}, Barrett \& Gisin~\cite{Barrett2011} and recently Thinh {\em et al.}~\cite{Thinh2013} studied how different possible relaxations influence well-known Bell-inequalities. Colbeck and Renner~\cite{Colbeck2012} introduced the idea of randomness amplification, in which a quantum protocol produces random outcomes even though complete free choice is not given. This was further developed by Gallego {\em et al.}~\cite{Gallego2013} and others~\cite{Mironowicz,Grudka,Ramanathan,Coudron2014,Miller2014}.

A common denominator of these works is that they study and use well-known Bell-inequalities. On the contrary, in this Letter, we derive Bell-like inequalities that are specifically suited for a measurement dependent scenario. Using these we show that quantum nonlocality allows for correlations that cannot be explained by any local models exploiting measurement dependence, even when the dependence is arbitrarily strong, as long as some measurement independence is retained (in the sense that we explain more precisely below).

\textit{Bell-locality.-} In a Bell test, two space-like separated parties, usually referred to as Alice and Bob, have access to two boxes. They can give these boxes an input, denoted by the random variables $X$ and $Y$ respectively, and each of the boxes gives back an outcome, $A$ and $B$ as depicted in Figure \ref{figmdlscenario}. In a quantum mechanical scenario, each box is given by a quantum system and the inputs determine measurements that are performed on this system. By performing many runs, the parties collect data that allows them to estimate with what probability a given input-pair $xy$ leads to an outcome-pair $ab$, i.e., they estimate the conditional probability distribution $P_{AB|XY}$. Note that we use capital letters for random variables and lower case letters for the values that the corresponding random variable can take.
The question a Bell test is trying to answer is whether these correlations could be explained by a (Bell-) local~\cite{Brunner:RMP} model, allowing for the existence of some underlying hidden strategy, denoted by $\Lambda$. We say that a correlation is  local if~\cite{Bell1964}
\begin{align}
\label{lassumption}
P(ab|xy)=\int\dd\lambda \rho(\lambda)P(a|x\lambda)P(b|y\lambda).
\end{align}
A correlation cannot be written as such an integral if and only if it violates a Bell inequality.

However, when performing an actual Bell test, an additional assumption has to be made: the inputs $X$ and $Y$ have to be chosen freely, i.e., uncorrelated to the hidden strategy $\Lambda$~\cite{Bell2004},
\begin{align}
P(xy|\lambda)=P(xy) \text{ } \forall\, x,y,\lambda
\end{align}
Following Hall~\citep{Hall2011}, Barrett and Gisin~\citep{Barrett2011}, we call this assumption \textit{measurement independence}.

\begin{figure}
\includegraphics[width=0.4\textwidth]{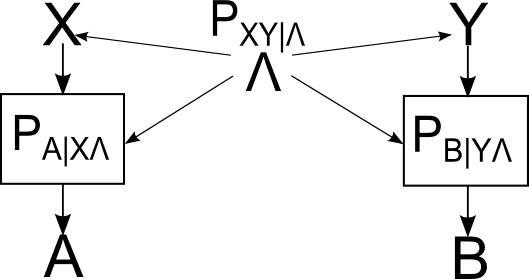}
\caption{Two boxes are programmed by a hidden common strategy $\Lambda$. This common strategy is also correlated with the inputs $X$ and $Y$ to the boxes. We call this measurement dependent locality.}
\label{figmdlscenario}
\end{figure}

\textit{Measurement dependence.-} We now analyse the case where complete measurement independence is not given. It is useful to consider, for this case, the full distribution $P_{ABXY}$, which --- contrary to the conditional distribution $P_{AB|XY}$ --- takes into account the distribution of the inputs $X$ and $Y$. We say that a correlation $P_{ABXY}$ is \textit{measurement dependent local} (MDL) if
\begin{align}
\label{mdlassumption}
P(abxy)=\int \dd\lambda \rho(\lambda)P(xy|\lambda)P(a|x\lambda)P(b|y\lambda).
\end{align}
As stated previously, if we allow measurement dependence and make no further assumptions, it is impossible to show that quantum mechanics is nonlocal. However, if one bounds the correlations between the inputs and the hidden strategy by imposing upper and lower bounds on the conditional distribution
\begin{align}
\label{lhassumption}
\ell\leq P(xy|\lambda)\leq h,
\end{align}
then interesting conclusions can be derived.
It is common to refer to such an assumption as a condition on the input-source~\cite{Colbeck2012}. Examples of such sources are the min-entropy sources~\cite{Thinh2013}, which have been studied in recent works~\cite{Bouda2014}.

We say that a correlation is \textit{measurement dependent nonlocal} for a given $\ell$ and $h$ if it cannot be expressed in the integral form given by (\ref{mdlassumption}) when assuming the lower and upper bounds coming from (\ref{lhassumption}).

\textit{The set of MDL-correlations.-} The set of measurement dependent local distributions for a given $\ell$ and $h$, i.e., the set of $P_{ABXY}$ satisfying (\ref{mdlassumption}) and (\ref{lhassumption}), turns out to be a convex set with a finite number of extremal points: a convex polytope~\SuppMat. The set can thus be fully characterized using a finite set of Bell-like inequalities: a distribution is measurement dependent nonlocal if and only if it violates at least one of these MDL-inequalities.

\textit{Quantum violation of a specific MDL-inequality.-} In the following, we focus on analysing the simplest possible nonlocality scenario: the inputs $X$ and $Y$ and outputs $A$ and $B$ of both parties are taken to be binary random variables, taking values $0$ or $1$. In terms of $P(abxy)$, one useful parametrized MDL-inequality, derived using the polytope structure of the MDL-set, is given by~\SuppMat
\begin{align}
\label{goldenineq}
\ell P(0000) - h\big(P(0101)+P(1010)+P(0011)\big)\stackrel{MDL}{\le} 0.
\end{align}
This inequality allows us to prove our main result.

\textit{Main result:} Quantum mechanics is measurement dependent nonlocal for any $\ell>0$ and for any $h$.

A state that exhibits this property is the 2-qubit state
\begin{align}
\label{goldenstate}
\left| Au\right. \rangle = \frac{1}{\sqrt{3}}\left(\frac{\sqrt{5}-1}{2}\left| 00\right. \rangle+\frac{\sqrt{5}+1}{2}\left|11 \right. \rangle\right),
\end{align}
on which Alice and Bob perform the rank 1 projective measurements defined by $\ket{A_0(\theta)}=\cos\theta\ket{0}+\sin\theta\ket{1}$, $\ket{A_1(\theta)}=\mid A_0(\theta-\frac{\pi}{4})\rangle$, $\ket{B_0(\theta)}=\ket{A_0(-\theta)}$ and $\ket{B_1(\theta)}=\ket{A_1(-\theta)}$ with  $\theta=\arccos{\sqrt{\frac{1}{2}+\frac{1}{\sqrt{5}}}}$.

Evaluating the left-hand side of inequality (\ref{goldenineq}) for this state and these measurements, we find $\ell\frac{1}{12}P_{XY}(00)$, where $P_{XY}(00)$ is the probability of choosing the inputs $x=0$ and $y=0$. This proves that measurement dependent local distributions cannot explain quantum correlations, as long as it is impossible for any hidden strategy to exclude the possibility that a certain input-pair occurs, i.e., $\ell>0$\footnote{In fact, any quantum correlation violating Hardy's paradox~\cite{Hardy93} violates inequality (\ref{goldenineq}).}. Note that this condition also excludes the possibility of having fully dependent inputs for one of the two parties since $P(x|\lambda)=0$ implies $P(xy|\lambda)=0$. A visual representation of inequality (\ref{goldenineq}) can be found in Fig. 2.

\textit{MDL-correlations satisfying additional physical constraints.-} Motivated by the idea that information needs a physical carrier, most nonlocality experiments are conducted under the assumption that no information can be transmitted between the parties by the use of such boxes, for example by performing the experiment in spacelike seperation. In other words, the input to Alice's box cannot influence the outcome on Bob's side and vice versa, i.e.,
\begin{align}
P(a|xy)&=P(a|xy')\text{ } \forall a,x,y,y'\nonumber\\
P(b|xy)&=P(b|x'y)\text{ } \forall b,x,x',y.
\label{nosignaling}
\end{align}
These are the nonsignaling assumptions~\cite{Popescu1994,Barrett05}. Nonsignaling, as opposed to measurement independence, can in principle be verified in a protocol by checking the equalities (\ref{nosignaling}).

The measurement dependent local correlations given by (\ref{mdlassumption}) are not inherently nonsignaling due to the hidden strategy $\Lambda$ establishing correlations between Alice's input $X$ and Bob's output $B$ and vice versa. Since, as stated above, these equalities can in principle be verified, we impose that they are satisfied in the following.

Additionally the experimenters can observe the input distribution $P_{XY}$ given by $P(xy)=\int\dd\lambda \rho(\lambda)P(xy|\lambda)$. Therefore any experiment or protocol involving Bell tests can make use of this knowledge. We consider, from here onwards, the case in which the input-distribution is observed to be uniform, meaning that 
\begin{align}
P(xy)=P(x'y') \text{ }\forall x,x',y,y'.
\label{pxyuniform}
\end{align}

\textit{Comparison with CHSH.-} Instead of using inequality (\ref{goldenineq}), one can try to show the measurement dependent nonlocality of quantum theory by using the well-known Clauser-Horne-Shimony-Holt (CHSH) expression~\cite{Clauser1969}
\begin{align}
\label{CHSH}
\text{CHSH}=\sum_{abxy}(-1)^{a+b+xy}P(ab|xy).
\end{align}
It is well-known that quantum mechanics respects Cirel'son's bound~\cite{Tsirelson}
\begin{align}
\text{CHSH}\stackrel{\mathcal{Q}}{\le}2\sqrt{2}.
\end{align}
Therefore, if for a given $\ell$ and $h$, the MDL-set given by (\ref{mdlassumption}) and (\ref{lhassumption}), allows for correlations with $\text{CHSH}\geq 2\sqrt{2}$, then the inequality cannot be used to reveal measurement dependent nonlocality.

Using the polytope structure of the MDL-set, we find that MDL-correlations, even with the additional constraints of nonsignaling (\ref{nosignaling}) and uniform inputs (\ref{pxyuniform}), can reach~\SuppMat
\begin{align}
\text{CHSH}=4(1-2\ell'),
\label{mdlchsh}
\end{align}
where $\ell'=\max(\ell,1-3h)$. Comparing this value to the quantum bound of $2\sqrt{2}$, we find that for $\ell'\leq \frac{2-\sqrt{2}}{4}$, it is impossible for CHSH to reveal the measurement dependent nonlocal behavior of quantum mechanics. Inequality (\ref{goldenineq}) on the other hand is able to reveal this $\forall \ell'>0$.

\begin{figure}[ht]
\includegraphics[width=0.4\textwidth]{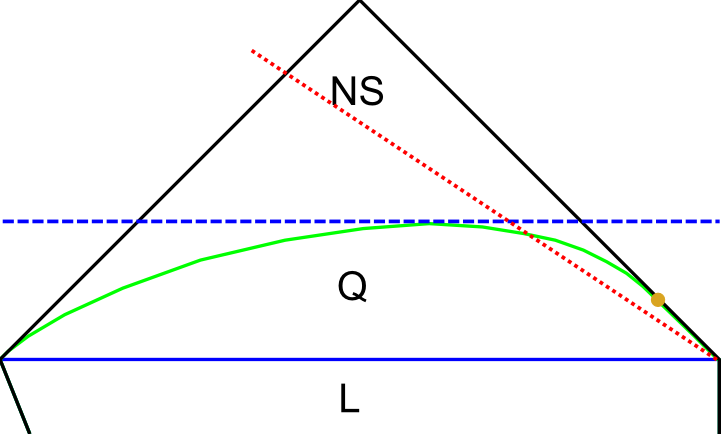}

\caption{(Color online) A 2-dimensional slice in the nonsignaling space. Below the solid horizontal (blue) line (CHSH) is the set of Bell-local distributions. The quantum set is delimited by the curved (green) line. The set of nonsignaling distributions lies below the black triangle. For $h=1-3\ell$, inequality (\ref{goldenineq}) tilts from CHSH ($\ell=\frac{1}{4}$) via the dotted (red) line ($\ell\approx \frac{2-\sqrt{2}}{4}$) to the nonsignaling border ($\ell=0$). The correlation given by measuring $\ket{Au}$ (golden dot) violates this inequality $\forall\ell>0$. We also display the CHSH inequality with an adjusted MDL bound for $\ell\approx \frac{2-\sqrt{2}}{4}$ (dashed blue line). It can be seen that the nonlocality of quantum mechanics cannot be shown using this inequality.}
\label{GoldenSlice}
\end{figure}

\textit{Input sources with $\ell=0$.- } One specific input-source is the min-entropy source~\cite{Thinh2013}. The conditional min-entropy is defined as
\begin{align}
H_{min}(XY|\lambda)=-\log_2\max_{x,y}P(xy|\lambda).
\end{align}
Using a min-entropy source means that $H_{min}(XY|\lambda)$ is lower bounded by some value $H$ $\forall\lambda$. In our language it corresponds to setting the lower bound $\ell=0$ and the upper bound $h=2^{-H}$ in condition (\ref{lhassumption}).

For $X$, $Y$, $A$, $B$ $\in \{0,1\}$ and specific values of $h$, we obtain a complete set of  MDL-inequalities. We say that a set of inequalities is complete if every measurement dependent nonlocal distribution violates at least one inequality in this set while measurement dependent local distributions, cf. \eqref{mdlassumption},  respect all inequalities. For example, in the corresponding Bell-locality scenario it is known that the CHSH-inequalities form a complete set.

A first observation to make is that maximal min-entropy, meaning that $h=\frac{1}{4}$ implies that $P_{XY|\Lambda=\lambda}$ is uniform $\forall\lambda$. This conclusion follows from the fact that every probability distribution has to be normalized, i.e., $\sum_{x,y}P(xy|\lambda)=1\text{  }\forall\lambda$, and that probabilities are non-negative. Since the inputs are not biased by $\lambda$, this corresponds to imposing measurement independence and is therefore equivalent to the standard Bell-locality. As already stated, the CHSH-inequalities form a complete set in this case.

\begin{table*}[ht]

\begin{adjustbox}{width=\textwidth,totalheight=\textheight,keepaspectratio}
\begin{tabular}{c|c|c|c|c|c|c|c|c}
\large
$1$&$P_{A|X}(0|0)$&$P_{A|X}(0|1)$&$P_{B|Y}(0|0)$&$P(00|00)$&$P(00|10)$&$P_{B|Y}(0|1)$&$P(00|01)$&$P(00|11)$\\
\hline
$12h^2-11h+2$&$2h-1$&$4h-1$&$2h-1$&$2h$&$2-6h$&$4h-1$&$2-6h$&$-2h$\\
$12h^2-11h+2$&$4h-1$&$3h-1$&$4h-1$&$-h$&$1-3h$&$3h-1$&$1-3h$&$1-3h$\\
$11h^2-8h+1$&$-4h^2+5h-1$&$5h^2-4h+1$&$-4h^2+5h-1$&$-3h^2-2h+1$&$3h^2-2h$&$5h^2-4h+1$&$3h^2-2h$&$-9h^2+9h-2$\\
$8h^2-7h+1$&$4h^2$&$0$&$-4h^2+5h-1$&$-h$&$1-3h$&$-4h^2+2h$&$-h$&$3h-1$\\
$13h^2-8h+1$&$-8h^2+6h-1$&$-5h^2+2h$&$-h^2+h$&$5h^2-2h$&$h^2-h$&$0$&$3h^2-4h+1$&$-3h^2+4h-1$\\
$20h^2-13h+2$&$-8h^2+6h-1$&$-7h^2+5h-1$&$-8h^2+6h-1$&$5h^2-2h$&$3h^2-4h+1$&$-7h^2+5h-1$&$3h^2-4h+1$&$-h^2+h$\\
$1-4h$&$3h-1$&$0$&$3h-1$&$1-3h$&$h$&$0$&$h$&$-h$\\
\end{tabular}
\end{adjustbox}

\normalsize
\caption{Conjectured families of MDL inequalities for $h\in\rbrack \frac{1}{4},\frac{1}{3}\lbrack$. The Table contains the coefficients belonging to each term (given in the first row). We denote by $P_{A|X}(a|x)$ the marginal distribution over Alice's ouput $A$ conditioned on her input $X$ and similarly for Bob. The expression being $\leq 0$ is a representative MDL inequality from each family.}
\label{conjineqsCG}
\end{table*}

Another special value is $h=\frac{1}{3}$. It turns out that if $h\geq\frac{1}{3}$, measurement dependent local correlations can reproduce any nonsignaling distributions. Since the set of quantum correlations is a strict subset of the set of non-signaling correlations, it is therefore impossible to see measurement dependent nonlocality in this case. The reason this does not occur for $h<\frac{1}{3}$ is due to the fact that for these values of $h$ the normalization and non-negativity of probabilities imply that no input-pair can be excluded, i.e., $P(xy|\lambda)>0$ $\forall x,y,\lambda$.

The interesting case is therefore $h\in]\frac{1}{4},\frac{1}{3}[$. For \emph{each fixed value} of $h$ in this interval, one can make use of a standard software package~\cite{Porta} to obtain the complete set of inequalities characterizing the set of MDL correlations. We performed this computation for several values of $h\in]\frac{1}{4},\frac{1}{3}[$. For each of these chosen values\SuppMat, we always found 7 families of inequalities, where we say that two inequalities belong to the same family if one can be obtained from the other by simply relabeling the inputs and outputs or by exchanging the roles of the two parties. As a function of $h$, the inequalities we found can be expressed as in Table \ref{conjineqsCG}. Based on the above observation, we conjecture that the inequalities of Table \ref{conjineqsCG} form a complete set for {\em all} $h\in]\frac{1}{4},\frac{1}{3}[$~\SuppMat. A visual representation of the evolution of the MDL-polytope as $h$ goes from $\frac{1}{4}$ to $\frac{1}{3}$ can be seen in Fig. 3.

It is interesting to note that the well-known CHSH-inequality is not among these 7 families. From (\ref{mdlchsh}), we see that for $h\geq\frac{2+\sqrt{2}}{12}\approx 0.2845$, quantum mechanics can no longer outperform the measurement dependent local correlations when looking at CHSH as given by (\ref{CHSH}). This was already shown by  Thinh \textit{et al}. ~\cite{Thinh2013}. CHSH is therefore only useful up to this critical value of $h$. On the other hand, all 7 families introduced in Table \ref{conjineqsCG} can be violated for values larger than $\frac{2+\sqrt{2}}{12}$. In fact, inequalities 6 and 7 can reveal quantum nonlocality for all $h$ below the critical value of $\frac{1}{3}$. This shows that the complete set presented here is better suited for the task of witnessing measurement dependent quantum nonlocality than CHSH.

\begin{figure}[ht]
\includegraphics[width=0.4\textwidth]{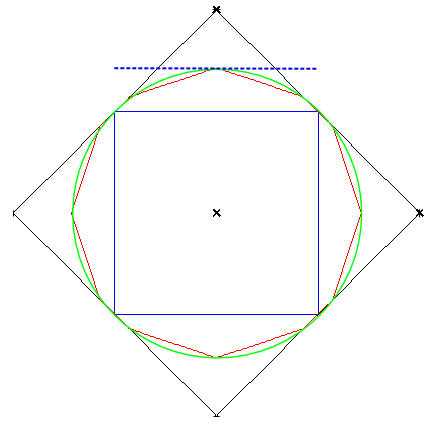}

\caption{(Color online) A different 2-dimensional slice through the nonsignaling space. We show the local polytope (inner square, blue), the nonsignaling polytope (outer, rotated square, black) as well as the MDL-polytope for a min-entropy-source with $h\approx \frac{2+\sqrt{2}}{12}$ (thin line between the squares, red). Since the local polytope corresponds to $h=\frac{1}{4}$ and the nonsignaling polytope to $h=\frac{1}{3}$, it can be seen how the MDL-polytope transforms as a function of $h$. The quantum set in this slice is bounded by the green circle. The CHSH inequality with an adjusted bound for $h\approx \frac{2+\sqrt{2}}{12}$ (dashed line, blue) cannot be used to reveal the nonlocality of quantum mechanics for $h\geq \frac{2+\sqrt{2}}{12}$.}
\label{UPRPRslice}
\end{figure}

\textit{Conclusions.-} Bell-locality, the essential concept when working in any kind of device independent scenario, includes the untestable assumption of measurement independence. We have analyzed what happens when this assumption is relaxed and found that, as with Bell-locality, it is sufficient to work with a finite number of  Bell-like inequalities. Using one such inequality, we showed that the nonlocality of quantum mechanics can be manifested as long as an arbitrarily small amount of free choice is guaranteed. Surprisingly, the simplest nontrivial scenario involving only two parties (and binary-outcome measurements) is already sufficient to arrive at this conclusion.

We have also presented inequalities that are better suited to measurement dependent scenarios than the CHSH-inequality. In fact, with the additional assumption of nonsignaling and uniform observed inputs, we obtained a set of Bell-like inequalities --- which we conjecture to be complete --- for the measurement dependent local set of two parties, two inputs (with a min-entropy input source), and two outputs.  In general, our observations that MDL correlations can be fully characterized using Bell-like inequalities provides a powerful framework for the study of measurement dependent quantum nonlocality and related applications. For instance, the MDL polytope presented in this Letter may become a useful tool for the analysis of tasks like randomness extraction\cite{Bouda2014} and amplification~\cite{Colbeck2012}. So far, inequalities suitable for such tasks were guessed or inspired by the local polytope. Inequalities derived or inspired from the MDL polytope should be better suited for the task. Specifically, it would be interesting to see whether the bipartite scenario with binary inputs and outputs could indeed be sufficient to perform a randomness amplification protocol using one of the inequalities presented here.

The framework introduced in this Letter allows one to study further other possible assumption on the input source. A natural possibility would be that any correlations between the random number generators that the two parties use to determine their inputs must come from a local hidden variable, i.e., $P(xy|\lambda)=P(x|\lambda)P(y|\lambda)$, a problem that we shall leave for future research.

\textit{Acknowledgements} We acknowledge helpful discussions with Valerio Scarani and Jean-Daniel Bancal as well as financial support from the European projects CHIST-ERA DIQIP, the COST Action MP1006, the ERC grant 258932,  the European Union Seventh Framework Programme via the RAQUEL project (grant agreement no 323970) and the Swiss project NCCR-QSIT.

\bibliography{Maintext}
\nocite{Gruenbaum67}

\newpage
\onecolumngrid
\appendix

\section{Introduction and notation}
In this supplementary material, we will prove theorems and lemmas necessary to establish the measurement dependent local inequalities that have been used in the maintext to show that the nonlocality exhibited by quantum mechanics cannot be reproduced using a measurement dependent local model.

We now introduce the notation used in this paper. We will only consider the two-party case, but the extension to more parties is immediate. We will denote a probability distribution over the random variable $Z$ by $P_{Z}$ and the probability that $Z$ takes value $z$ by $P_Z(z)$. We will often omit the random variable and just write $P(z)$ if the random variable is clear by context.

The setup considered can be seen in Figure (1) in the maintext. We consider two parties, Alice and Bob, that are spacelike seperated. They have access to two boxes. The boxes each take an input, denoted by $X$ and $Y$ respectively and return an output, $A$ and $B$ respectively. The boxes (and thus their outputs) as well as the inputs can all be correlated with a common strategy $\Lambda$. The relevant quantity that concerns us is $P_{ABXY}$ with
\begin{align}
\label{pabxy}
P(abxy)=\int\dd\lambda\rho(\lambda)P(xy|\lambda)P(ab|xy\lambda).
\end{align}

\section{The measurement dependent local polytope}
\label{mdlpolytope}
In this Letter, we analyse the case where the outputs of the boxes are classically determined by the respective input and a possible common strategy and the amount of correlation between the inputs and the common strategy is bounded. We find that the possible resulting probability distributions form a polytope.

We denote by $\mathcal{P}_Z$ the space of probability distributions over the random variable $Z$ and by $\mathcal{P}_{Z|W}$ the space of conditional probability distributions over the random variable $Z$ conditioned on the random variable $W$.

Let $X, Y, A, B$ be random variables with alphabetsize $n_X,n_Y,n_A,n_B$ respectively.
We define the local polytope
\begin{align}
\label{localpolytope}
\mathcal{L}=\big\{P_{AB|XY}\in\mathcal{P}_{AB|XY}: P(ab&|xy)=\int\dd\lambda\rho(\lambda)P(a|x\lambda)P(b|y\lambda)\nonumber\\
&P_{A|X\Lambda}\in\mathcal{P}_{A|X\Lambda},P_{B|Y\Lambda}\in\mathcal{P}_{B|Y\Lambda}\nonumber\\
&\rho(\lambda)\geq 0\,\forall \,\lambda, \int\dd\lambda\rho(\lambda)=1\big\}
\end{align}
which has the so-called deterministic points as its vertices
\begin{align}
\label{localvertices}
\mathcal{V}_{\mathcal{L}}=\big\{V_{AB|XY}\in\mathcal{P}_{AB|XY} : V(ab|xy)=V(a|x)V(b|y) \text{ with }&V_{A|X}\in\mathcal{P}_{A|X}, V_{B|Y}\in\mathcal{P}_{B|Y} \nonumber\\
&V(a|x)\in\{0,1\}, V(b|y)\in\{0,1\}\big\}.
\end{align}

For $\ell,h\in[0,1]$ s.t. 
\begin{align}
0\leq\ell\leq\frac{1}{n_Xn_Y}\leq h\leq 1,
\label{lhcond}
\end{align}

we define the input polytope
\begin{align}
\label{inputpolytope}
\mathcal{I}(\ell,h)=\big\{P_{XY}\in\mathcal{P}_{XY}: \ell\leq P(xy)\leq h\, \forall\, x,y\big\}.
\end{align}
The fact that $\mathcal{I}(\ell,h)$ is a convex polytope is proven by lemma \ref{pxylemma}. If $\ell=\frac{1}{n_Xn_Y}$ or $h=\frac{1}{n_Xn_Y}$, the polytope consists of only one point given by the uniform distribution $P(xy)=\frac{1}{n_Xn_Y}$ $\forall x,y$. For $\ell<\frac{1}{n_Xn_Y}<h$, let $n=\lfloor\frac{1-n_Xn_Y\ell}{h-\ell}\rfloor$,
\begin{align}
\label{inputlist}
S(\ell,h)=(\underbrace{h,\cdots,h}_{n \text{ times}},\underbrace{\ell,\cdots,\ell}_{(n_Xn_Y-n-1)\text{ times}},1-nh-(n_Xn_Y-n-1)\ell).
\end{align}
and $\Pi_{n_Xn_Y}$ be the set of all permutations of $n_Xn_Y$ elements. Then the set of vertices of $\mathcal{I}(\ell,h)$ is given by
\begin{align}
\label{inputvertices}
\mathcal{V}_{\mathcal{I}}(\ell,h)=\big\{V_{XY}\in\mathcal{P}_{XY}: \exists\pi\in\Pi_{n_Xn_Y}{ s.t. }V_{XY}=\pi(S)\big\}.
\end{align}

\begin{thm}
Let $\mathcal{L}$, $\mathcal{V}_{\mathcal{L}}$, $\mathcal{I}(\ell,h)$ and $\mathcal{V}_{\mathcal{I}}(\ell,h)$ be defined as in (\ref{localpolytope}), (\ref{localvertices}), (\ref{inputpolytope}) and (\ref{inputvertices}).
For $\ell,h\in[0,1]$ fulfilling condition (\ref{lhcond}), define
\begin{align*}
\mathcal{MDL}(\ell,h)=\big\{P_{ABXY}\in\mathcal{P}_{ABXY} :  &P(abxy)=\int\dd\lambda\rho(\lambda)P(xy|\lambda)P(ab|xy\lambda),\\
&P_{XY|\Lambda=\lambda}\in\mathcal{I}(\ell,h)\text{, }P_{AB|XY\Lambda=\lambda}\in\mathcal{L}\text{ }\text{ }\forall\lambda\\
&\rho(\lambda)\geq 0 \forall\lambda, \int d\lambda\rho(\lambda)=1\big\}.
\end{align*}
This is a polytope and the set of vertices is a subset of
\begin{align}\label{Eq:MDLVertices}
\mathcal{V}_{\mathcal{MDL}}(\ell,h)=\big\{V_{ABXY}\in\mathcal{P}_{ABXY}: V(abxy)=V(xy)V(ab|xy)\text{ with }&V_{XY}\in\mathcal{V}_{\mathcal{I}}(\ell,h)\text{ and }V_{AB|XY}\in \mathcal{V}_{\mathcal{L}}\big\}.
\end{align}
\end{thm}

\begin{proof}
The correlations in $\mathcal{MDL}(\ell,h)$ are given by
\begin{align*}
P(abxy)=\int\dd\lambda\rho(\lambda)P(xy|\lambda)P(ab|xy\lambda).
\end{align*}
with $P_{AB|XY\Lambda=\lambda}$ and $P_{XY|\Lambda=\lambda}$ each being chosen from a polytope. Hence Theorem~\ref{polytopethm} applies and $\mathcal{MDL}(\ell,h)$ is a convex polytope with vertices  given by the Cartesian product of the vertices of the two constituent polytopes.
\end{proof}

Since $\mathcal{MDL}(\ell,h)$ is a convex polytope, the maximal and minimal value of any linear expression of probabilities is achieved by one of the vertices. The validity of the inequalities presented in the main text can hence be verified by checking that they hold for all the points in $\mathcal{V}_{\mathcal{MDL}}(\ell,h)$.

\subsection{The MDL bound for the CHSH expression}

The bound for CHSH cannot be derived in this way, since we impose the conditions of nonsignaling
\begin{align}
\sum_bP(ab|xy)&=\sum_bP(ab|xy')\text{ } \forall a,x,y,y'\nonumber\\
\sum_aP(ab|xy)&=\sum_aP(ab|x'y)\text{ } \forall a,x,x',y
\label{nosignaling-app}
\end{align}
and uniform inputs
\begin{align}
P(xy)=P(x'y')=\frac{1}{4} \text{ }\forall x,x',y,y'.
\label{pxyuniform-app}
\end{align}
on top of the constraints of measurement dependent locality. The bound given in the main text was found by finding an explicit convex combination of the vertices that additionally fulfilled the conditions (\ref{nosignaling-app}) and (\ref{pxyuniform-app}).

First, we write the CHSH expression using full probabilities as
\begin{align*}
\text{CHSH}_{\text{full}}=\sum_{abxy}(-1)^{a+b+xy}P(abxy).
\end{align*}
Using the vertices of the $\mathcal{MDL}(\ell,h)$, we find that there are 8, 24 or 48 MDL-vertices\footnote{The number depends on the values of $\ell$ and $h$. If the vertices of the input-polytope are permutations of $(h,h,h,1-3h)$ then $1-2\ell'$ is reached by 8 MDL-vertices, for $(h,h,1-2h-\ell,\ell)$ and $(1-3\ell,\ell,\ell,\ell)$ there are 24 and for $(h,1-h-2\ell,\ell,\ell)$ there are 48.} which all reach the maximum value of $1-2\ell'$ for $\text{CHSH}_{\text{full}}$ with $\ell'=\max(1-3h,\ell)$. A uniform mixture of these vertices therefore also has a $\text{CHSH}_{\text{full}}$-value of $1-2\ell'$. It turns out that additionally this uniform mixture is a nonsignaling distribution with $P(xy)=\frac{1}{4}$. Reverting back to conditional probability distributions, we find that this point has a CHSH-value of $4(1-2\ell')$.

\subsection{The MDL polytope assuming Eq.~\eqref{nosignaling-app} and Eq.~\eqref{pxyuniform-app}}

Any convex polytope can be seen as an intersection of finitely many half spaces, in other words by a finite set of linear inequalities. Given the vertices, we can use existing software packages to find these inequalities numerically for fixed values of $\ell$ and $h$. For the case of $n_X=n_Y=n_A=n_B=2$, $h\in]\frac{1}{4},\frac{1}{3}[$ and $\ell=0$, we performed this computation for $h=\frac{2}{7}, \frac{3}{8}, \frac{3}{11}, \frac{3}{10}, \frac{4}{15}, \frac{4}{13}, \frac{5}{16}, \frac{5}{17}$. We then imposed the conditions of nonsignaling (\ref{nosignaling-app}) and uniform inputs (\ref{pxyuniform-app}). These are linear equality constraints and a simple variable elimination in the inequalities that were found numerically is enough to enforce them. Regardless of the value of $h$, we found $8$ families of inequalities, where we say that $2$ inequalities belong to the same family if they can be transformed into each other by a relabeling of the inputs, outputs and/or parties. One of these families corresponds to the fact that probabilities are non-negative and is therefore not interesting for any analysis of physical theories. We conjectured the dependence on $h$ of the other $7$ families. They can be found in Table I in the main text.\footnote{As a consistency check, we have used the explicit characterization of the MDL polytope provided in Eq.~\eqref{Eq:MDLVertices} and verified --- using linear programming --- that for $10^4$ randomly generated values of $h\in]\frac{1}{4},\frac{1}{3}[$, no convex combination of the 64 extreme points satisfying Eq.~\eqref{nosignaling-app} and Eq.~\eqref{pxyuniform-app} can violate the 7 families of inequalities found. Moreover, in all these cases, each inequality can be saturated to, at least, a numerical precision of $7\times10^{-9}$. }

\section{Theorem on combining polytopes}
\label{polytopethmproof}
In this section, we will prove that if the input-distributions $P_{XY|\Lambda=\lambda}$ and output-distributions $P_{AB|XY\Lambda=\lambda}$ are chosen from a polytope for all $\lambda$, then the possible resulting $P_{ABXY}$ given by (\ref{pabxy}) also form a polytope.

\begin{thm}
\label{polytopethm}
Let $\mathcal{I}_{XY}$ and $\mathcal{O}_{AB|XY}$ both be polytopes with their respective vertices being $\left\{V_{XY}^{\lambda'}\right\}_{\lambda'}$ and $\left\{V_{AB|XY}^{\lambda''}\right\}_{\lambda''}$. Then
\begin{align*}
\mathcal{R}_{ABXY} = \big\{P&_{ABXY} :P(abxy)=\int d\lambda \rho(\lambda)P(xy|\lambda)P(ab|xy\lambda)\\
&P_{XY|\Lambda=\lambda}\in \mathcal{I}_{XY}, P_{AB|XY\Lambda=\lambda}\in \mathcal{O}_{AB|XY}, \rho(\lambda)\geq 0 \forall\lambda, \int d\lambda\rho(\lambda)=1\big\}
\end{align*}
is a polytope and its vertices are a subset of
\begin{align*}
\mathcal{V}_{ABXY}=\left\{V_{ABXY}^{(\lambda'\lambda'')} : V^{(\lambda'\lambda'')}(abxy)=V^{\lambda'}(xy)V^{\lambda''}(ab|xy)\right\}_{(\lambda'\lambda'')}.
\end{align*}
\end{thm}

\begin{proof}
The proof consists of 3 steps:

\textit{Step 1} We show that $\mathcal{V}_{ABXY}\subset\mathcal{R}_{ABXY}$.\\
This is clear since $V_{XY}^{\lambda'}\in\mathcal{I}_{XY}$ and $V_{AB|XY}^{\lambda''}\in\mathcal{O}_{AB|XY}$.

\textit{Step 2} We show that every convex combination of elements in $\mathcal{V}_{ABXY}$ is in $\mathcal{R}_{ABXY}$
\begin{align*}
\sum_{(\lambda'\lambda'')}\alpha_{(\lambda'\lambda'')}V_{ABXY}^{(\lambda'\lambda'')} \in \mathcal{R}_{ABXY}.
\end{align*}

Since $V_{XY}^{\lambda'}\in\mathcal{I}_{XY}$ and $V_{AB|XY}^{\lambda''}\in\mathcal{O}_{AB|XY}$, we define 
\begin{align*}
P_{XY|\Lambda=(\lambda'\lambda'')}&=V_{XY}^{\lambda'}\\
P_{AB|XY\Lambda=(\lambda'\lambda'')}&=V_{AB|XY}^{\lambda''}\\
\rho(\lambda)&=\sum_{(\lambda'\lambda'')}\delta(\lambda-(\lambda'\lambda''))\alpha_{(\lambda'\lambda'')}.
\end{align*}

And therefore
\begin{align*}
\sum_{(\lambda'\lambda'')}\alpha_{(\lambda'\lambda'')}V_{ABXY}^{(\lambda'\lambda'')} \in \mathcal{R}_{ABXY}.
\end{align*}

\textit{Step 3} We show that every $P_{ABXY}\in\mathcal{R}_{ABXY}$ can be written as a convex combination of the $V_{ABXY}\in\mathcal{V}_{ABXY}$. 
By definition we can write $\forall P_{ABXY}\in \mathcal{R}_{ABXY}$
\begin{equation}\label{Eq:PABXY}
\begin{split}
P_{ABXY}&=\int\dd\lambda \rho(\lambda) P_{ABXY}^\lambda, \rho(\lambda)\geq 0, \int\dd\lambda\rho(\lambda) = 1\\
P_{ABXY}^\lambda(abxy)&=\sum_{x'x''y'y''}\Phi^{xx'x''yy'y''}P_{XY}^{\lambda}(x'y') P_{AB|XY}(abx''y'')\\
\Phi^{xx'x''yy'y''}&=\delta^{xx'}\delta^{xx''}\delta^{yy'}\delta^{yy''}
\end{split}
\end{equation}
where $P_{XY}^\lambda\in\mathcal{I}_{XY}$ and $P_{AB|XY}\in\mathcal{O}_{AB|XY}$. Since $\mathcal{I}_{XY}$ and $\mathcal{O}_{AB|XY}$ are polytopes, their elements can be written as convex combination of the corresponding extreme points, i.e.,
\begin{equation}\label{Eq:ConvexDecomposition}
\begin{split}
P_{XY}^\lambda &= \sum_{\lambda'}i_{\lambda'}^\lambda V_{XY}^{\lambda'}\\
P_{AB|XY}^\lambda &= \sum_{\lambda''}o_{\lambda''}^\lambda V_{AB|XY}^{\lambda''},
\end{split}
\end{equation}
where $\sum_{\lambda'} i_{\lambda'}^\lambda=\sum_{\lambda''}o_{\lambda''}^\lambda=1$.
Putting Eqs.~\eqref{Eq:PABXY}-\eqref{Eq:ConvexDecomposition} together,  we get that for  $P_{ABXY}\in\mathcal{R}_{ABXY}$
\begin{align*}
P_{ABXY}(abxy) &= \int \dd\lambda\rho(\lambda) \sum_{\lambda',\lambda''}i_{\lambda'}^\lambda o_{\lambda''}^\lambda \sum_{x'x''y'y''}\Phi^{xx'x''yy'y''}V_{XY}^{\lambda'}(x'y') V_{AB|XY}^{\lambda''}(abx''y'')\\
&= \sum_{\lambda',\lambda''}\varphi_{(\lambda'\lambda'')}V^{(\lambda'\lambda'')}\\
V^{(\lambda'\lambda'')} &= \Phi^{xx'x''yy'y''}V_{XY}^{\lambda'}(x'y') V_{AB|XY}^{\lambda''}(abx''y'')\\
\varphi_{(\lambda'\lambda'')} &= \int \dd\lambda\rho(\lambda) i_{\lambda'}^\lambda o_{\lambda''}^\lambda,
\end{align*}
Since $\varphi_{(\lambda'\lambda'')}\geq 0$ $\forall\lambda',\lambda''$ and $\sum_{\lambda',\lambda''}\varphi_{(\lambda'\lambda'')} = 1$, we  have thus shown that every $P_{ABXY}\in\mathcal{R}_{ABXY}$ can be written as a convex combination of the $V_{ABXY}^{(\lambda'\lambda'')}$.

This proves the theorem.
\end{proof}

\section{Lemma on the considered input-distributions}
\label{pxylemmaproof}

\begin{lma}
\label{pxylemma}
Let $\mathcal{I}(\ell,h)$ be defined as in (\ref{inputpolytope}) for $\ell,h\in[0,1]$ fulfilling the condition (\ref{lhcond}). This is a polytope. For $\ell=\frac{1}{n_Xn_Y}$ or $h=\frac{1}{n_Xn_Y}$, the only vertex is $P(xy)=\frac{1}{n_Xn_Y}$ $\forall x,y$. Otherwise the set of vertices is given by $\mathcal{V}_{\mathcal{I}}(\ell,h)$ as defined in (\ref{inputvertices}).
\end{lma}

\begin{proof}

The fact that probability distributions are normalized, i.e.,
\begin{align*}
\sum_{x,y}P(xy)=1  \text{ }\forall\, P_{XY}\in\mathcal{P}_{XY}
\end{align*}
implies that $\mathcal{I}(\frac{1}{n_Xn_Y},h)=\mathcal{I}(\ell,\frac{1}{n_Xn_Y})$ consists of only one point: $P(xy)=\frac{1}{n_Xn_Y}$. This proves the lemma for this case.

For $\ell<\frac{1}{n_Xn_Y}<h$, we note first that $\mathcal{I}(\ell,h)$ is a convex polytope since it is defined by linear constraints. These linear constraints are
\begin{subequations}\label{Eq:PolytopeConstraints}
\begin{align}
P(xy)&\geq \ell \text{ }\forall\, x,y\\
P(xy)&\leq h \text{ }\forall\, x,y\\
\sum_{xy}P(xy)&=1.\label{Eq:Normalization}
\end{align}
\end{subequations}
Due to the equality constraint, Eq.~\eqref{Eq:Normalization}, the dimension of the polytope is $n_Xn_Y-1$. Thus, $V$ is a vertex of $\mathcal{I}(\ell,h)$ if and only if it saturates at least $n_Xn_Y-1$ of the inequalities in Eq.~\eqref{Eq:PolytopeConstraints} \cite{Gruenbaum67}. Therefore we find that every vertex, when written as a vector in an $n_Xn_Y$-dimensional space, is a permutation of
\begin{align*}
S_n(\ell,h)=(\underbrace{h,\cdots,h}_{n \text{ times}},\underbrace{\ell,\cdots,\ell}_{(n_Xn_Y-n-1)\text{ times}},f),
\end{align*}
where $n$ is a positive integer and we have $\ell\leq f\leq h$,  $nh+(n_Xn_Y-n-1)\ell+f=1$ due to the polytope constraints in Eq.~\eqref{Eq:PolytopeConstraints}. The equality constraint implies
\begin{align}
\label{f}
f=1-nh-(n_Xn_Y-n-1)\ell.
\end{align}
It remains to show that the only possible value for $n$ is $n=\lfloor\frac{1-n_Xn_Y\ell}{h-\ell}\rfloor$.
Replacing $f$ in the inequality constraint by using (\ref{f}) we find
\begin{align*}
\ell\leq 1-nh-(n_Xn_Y-n-1)\ell\leq h\\
\Leftrightarrow \frac{1-n_Xn_Y\ell}{h-\ell}-1\leq n\leq \frac{1-n_Xn_Y\ell}{h-\ell}.
\end{align*}
Note that since $\ell<\frac{1}{n_Xn_Y}<h$, we have that $\frac{1-n_Xn_Y\ell}{h-\ell}>0$. We distinguish two cases:

$\frac{1-n_Xn_Y\ell}{h-\ell}\notin\mathbb{N}$: Then $n$ is an integer that lies between 2 non-integer real numbers whose difference is $1$. Therefore we get that $n=\lfloor\frac{1-n_Xn_Y\ell}{h-\ell}\rfloor=\lceil\frac{1-n_Xn_Y\ell}{h-\ell}-1\rceil$. For this case the lemma is proven.

$\frac{1-n_Xn_Y\ell}{h-\ell}\in\mathbb{N}$: In this case, both $n=\frac{1-n_Xn_Y\ell}{h-\ell}$ and $n=\frac{1-n_Xn_Y\ell}{h-\ell}-1$ are valid solutions. However, we find that in the first case, Eq. (\ref{f}) implies $f=\ell$ while in the second case it implies $f=h$. Therefore $S_{\frac{1-n_Xn_Y\ell}{h-\ell}}(\ell,h)=S_{\frac{1-n_Xn_Y\ell}{h-\ell}-1}(\ell,h)$ and both solutions yield the same set of vertices. We can thus set $n=\frac{1-n_Xn_Y\ell}{h-\ell}=\lfloor\frac{1-n_Xn_Y\ell}{h-\ell}\rfloor$.

This proves the lemma.

\end{proof}

\end{document}